\RequirePackage{amsmath} 
\documentclass[runningheads]{llncs}
\usepackage[T1]{fontenc}
\usepackage{mathtools}
\usepackage{amsmath}
\usepackage{amssymb}
\usepackage{booktabs}
\usepackage{cite}
\usepackage{array}
\usepackage{mathbbol}
\usepackage{adjustbox}
\usepackage{upgreek}
\usepackage[disable]{todonotes}
\usepackage{framed}
\usepackage{tikz-cd}
\usepackage[inline]{enumitem}
\usepackage{xcolor}
\usepackage{ulem}
\usepackage{lineno}
\usepackage{caption}
\usepackage{orcidlink}
\usepackage{mdframed}
\usepackage{tabularx}
\usepackage{hyperref}

\usepackage{color}

\usepackage{cleveref} 

\newcommand{\lhs}{\mbox{\texttt{lhs}}}
\newcommand{\rhs}{\mbox{\texttt{rhs}}}
\newcommand{\ord}{\precsim}

\newcommand{\cF}{\mathcal{F}}

\newcommand{\cV}{\mathcal{V}}

\newcommand{\cR}{\mathcal{R}}
\newcommand{\cE}{\mathcal{E}}

\newcommand{\qOmega}{\mathbb{\Omega}}
\newcommand{\qTheta}{\mathbb{\Theta}}
\newcommand{\qPhi}{\mathbb{\Phi}}

\newcommand{\Id}{\mathit{Id}}

\newcommand{\dom}{\mathit{dom}}
\newcommand{\ran}{\mathit{ran}}
\newcommand{\var}{\cV}
\newcommand{\pos}{\mathit{Pos}}
\renewcommand{\pos}{\mathit{O}}
\newcommand{\vpos}{\pos_{\cal V}}
\newcommand{\nvpos}{\pos_{\cal F}}
\newcommand{\bpos}{\pos_{\cal B}}

\newcommand{\mgu}{\ensuremath{\mathrm{mgu}}}

\newcommand{\sset}{\mathit{Sub}}

\newcommand{\narrow}{{\textup{\textsc{BQNarrow}}}}
\newcommand{\true}{{\textup{\textsc{true}}}}

\newcommand{\rulefont}[1]{\textsf{\textup{#1}}}     
\newcommand{\qetrulefont}[1]{\textsf{\textup{#1}}}  

\newcommand{\bfone}{\mathbb{1}}
\newcommand{\grdeg}{\partial}

\newcommand\newto[1][1.2em]{
	\mathrel{\kern-1pt\tikz[baseline=-0.5ex, 
	shorten <=1pt, shorten >=1pt] 
    \draw[-latex] (0,0) -- (#1,0);\kern-1pt}}

\newcommand{\newgets}[1][1.2em]{
	\mathrel{\kern-1pt\tikz[baseline=-0.5ex, 
	shorten <=1pt, shorten >=1pt] 
    \draw[-latex] (#1,0) -- (0,0);\kern-1pt}}

\newcommand\newlra[1][1.4em]{
	\mathrel{\tikz[baseline=-0.5ex, 
	shorten <=1pt, shorten >=1pt] 
    \draw[latex-latex] (0,0) -- (#1,0);}}

\newcommand\newda[1][1.0em]{
	\mathrel{\kern2pt\tikz[baseline=.5ex, 
	shorten <=1pt, shorten >=0pt] 
	\draw[-latex] (0,#1) -- (0,0);\kern2pt}}

\let\oldto\to

\renewcommand{\to}{\newto}
\renewcommand{\gets}{\newgets}
\renewcommand{\leftrightarrow}{\newlra}
\renewcommand{\downarrow}{\newda}

\newtheorem{notation}{Notation}

\makeatletter
\newcommand\notsotiny{\@setfontsize\notsotiny\@vipt\@viipt}
\makeatother

\allowdisplaybreaks

\usepackage{graphicx}
\begin{document}

\hyphenation{uni-fi-ca-tion}
\title{Graded Quantitative Narrowing}
%
%
\author{Mauricio Ayala-Rincón\inst{1}\orcidlink{0000-0003-0089-3905} \and Thaynara Arielly de Lima\inst{2}\orcidlink{0000-0002-0852-9086} \and
Georg Ehling\inst{3}\orcidlink{0009-0003-5931-9673} \and
Temur Kutsia\inst{3}\orcidlink{0000-0003-4084-7380}}
\authorrunning{M. Ayala-Rincón et al.}
%
\institute{Universidade de Brasília, Brazil\\
\email{ayala@unb.br}
\and
Universidade Federal de Goiás, Brazil\\
\email{thaynaradelima@ufg.br}\\
\and
Johannes Kepler Universit\"at Linz, Austria\\
\email{\{gehling,kutsia\}@risc.jku.at}}
\maketitle              
\begin{abstract}
The recently introduced framework of Graded Quantitative Rewriting is an innovative extension of traditional rewriting systems, in which rules are annotated with degrees drawn from a quantale. This framework provides a robust foundation for equational reasoning that incorporates metric aspects, such as the proximity between terms and the complexity of rewriting-based computations. Quantitative narrowing, introduced in this paper, generalizes quantitative rewriting by replacing matching with unification in reduction steps, enabling the reduction of terms even when they contain variables, through simultaneous instantiation and rewriting. In the standard (non-quantitative) setting, narrowing has been successfully applied in various domains, including functional logic programming, theorem proving, and equational unification. Here, we focus on quantitative narrowing to solve unification problems in quantitative equational theories over Lawverean quantales. We establish its soundness and discuss conditions under which completeness can be ensured. This approach allows us to solve quantitative equations in richer theories than those addressed by previous methods.
\end{abstract}

\section{Introduction}

While the notion of equality is fundamental to much of mathematics, it has also been viewed as limited in its ability to capture real-world phenomena, where entities are rarely ``exactly'' identical. As a result, various approaches have been developed to extend the equality predicate into a relation that accounts for proximity or similarity rather than strict identity. For these extensions, fundamental symbolic computation techniques such as, e.g., fuzzy equation solving \cite{DBLP:journals/fss/Ait-KaciP20,DBLP:journals/tcs/Sessa02,DBLP:conf/fscd/DunduaKMP20,Julian-Iranzo.2015,Pau.2022,Pau.2021,DBLP:journals/fuin/FormatoGS00} or metric-based equational reasoning \cite{Mardare.2016,Bacci.2020,DBLP:conf/lics/MardarePP17} have been investigated. 
A recent proposal by Gavazzo and Di Florio \cite{Gavazzo.2022} for quantitative equational theories is one of the more general approaches, containing some other frameworks as special cases. In this framework, equations are endowed with an element of a Lawverean quantale \cite{Lawvere1973}, measuring, in one way or another, the degree to which they hold. A unification algorithm for solving equations in a class of equational theories in this framework was introduced in~\cite{DBLP:conf/ijcar/EhlingK24}.

In the standard (non-quantitative) setting, narrowing was introduced as a generalization of rewriting that reduces terms containing variables by replacing matching with unification, enabling variable instantiation during reduction. Among its many applications, equational unification holds one of the central roles. Seminal works by Fay and Hullot, who introduced, respectively, ``ordinary'' and ``basic'' narrowing, proved the completeness of narrowing for solving equations in equational theories presented by convergent (confluent and terminating) term rewriting systems (TRSs)~\cite{Fay1979, Hullot1980}. Hullot's initial result on the termination of basic narrowing states that it terminates for convergent term rewriting systems, provided that every basic narrowing derivation starting from the right-hand side of any rewriting rule terminates. This result was later refined by Alpuente et al., who relaxed the requirement of convergence~\cite{DBLP:journals/tcs/AlpuenteEI09, DBLP:journals/igpl/AlpuenteEI11}. The completeness of ordinary narrowing holds even without the termination assumption, provided that the set of solutions is restricted to normalizable substitutions \cite{DBLP:conf/alp/MiddeldorpH92,DBLP:journals/aaecc/MiddeldorpH94}. However, this is not the case for basic narrowing, where the completeness fails even for normalizable substitutions if termination is not assumed. Narrowing has been extended to conditional rewriting theories and applied as a key mechanism in combining functional and logic programming \cite{DBLP:conf/aaai/DershowitzS88, DBLP:journals/entcs/EscobarMT07, DBLP:journals/cacm/AntoyH10}, and further implemented in deductive frameworks based on rewriting and logic (e.g., \cite{Hanus16Curry,DBLP:conf/birthday/ClavelDEELMT15}). It has been applied in protocol security verification \cite{DBLP:journals/jlap/LopezRuedaES23}, verification of termination of rewriting modulo equational axioms, and protocol security \cite{DBLP:journals/jlp/EscobarSM12, DBLP:journals/peerj-cs/TranDEO23}. Other applications are related to model checking \cite{DBLP:conf/rta/BaeEM13}, formalization of concurrent programming, and, in general, reachability analysis in deductive frameworks \cite{DBLP:journals/scp/CholewaEM15}.   

Graded quantitative rewriting offers a natural and fine-grained framework for formalizing the complexity of rewriting-based computations. Related approaches, such as weighted rewriting \cite{DBLP:conf/fscd/AhrensKGK25, AvanziniYFroCoS25}, typically focus on measuring reduction steps to analyze the complexity of rewriting derivations. Graded quantitative rewriting, in addition, provides a precise account of the cost of redex selection, a fundamental aspect of rewriting-based computation. This additional measure reflects the sensitivity to the position of the context in which quantitative redexes are identified and reductions are applied. A similar consideration arises in formal frameworks such as the lambda calculus, explicit substitutions, and nominal rewriting, where the cost of locating redexes to contract plays a crucial role in accurately modeling computational complexity.

Advances in quantitative computational techniques play a crucial role in automating reasoning about structures concerning various behavioral metrics. Quantitative narrowing presents a powerful mechanism that can be integrated into proof assistants to enhance automation in solving approximate equational problems, thereby supporting rigorous reasoning about the inherent imprecision in computational models of real-world systems. This paper takes a first step in this direction by introducing a graded quantitative narrowing calculus, where quantitative information is modeled using quantales. This is a general framework that can accommodate various well-known quantitative theories as special cases.

\paragraph{Contributions.} This paper introduces quantitative narrowing, defining ordinary and basic variants, and presents a rule-based basic narrowing procedure for solving equational unification problems over Lawverean quantales. This work extends the class of theories considered in quantitative unification as studied in~\cite{DBLP:conf/ijcar/EhlingK24}. We prove the soundness of the procedure and specify conditions under which completeness can be guaranteed, highlighting key differences from the well-known properties of ordinary and basic narrowing in the crisp first-order setting. By addressing unification problems over quantales using narrowing techniques, this paper takes a step toward establishing a framework for quantitative unification modulo equational theories.

\paragraph{Organization.} Section \ref{sect:prelim} introduces the necessary background on quantales, while Section \ref{sect:gradedTRS} defines graded quantitative equational and rewriting systems and introduces the quantitative unification problem. Section \ref{sect:narrowing} presents the notion of graded narrowing and introduces a set of basic narrowing inference rules, which are shown to be sound. Section \ref{sect:completeness} then discusses conditions for completeness. Finally, Section \ref{sect:conclusion} concludes the paper and briefly outlines directions for future work.

\section{Preliminaries}
\label{sect:prelim}

The basic notions and terminology on quantales are included in the following subsections.  

\subsection{Quantales} 
\label{subsect:quantales}

For the notions in this part, we follow  \cite{Gavazzo.2022,DBLP:journals/pacmpl/GavazzoF23}.
\begin{definition}[Quantale]
     A (unital) \emph{quantale} $\qOmega = (\Omega, \ord, \otimes, \kappa)$ consists of a monoid $(\Omega, \kappa, \otimes)$ and a complete lattice $(\Omega, \ord)$ (with join $\vee$ and meet $\wedge$) satisfying the following distributivity laws: $\delta \otimes \left( \bigvee_{i\in I} \varepsilon_i \right) =  \bigvee_{i\in I} (\delta \otimes \varepsilon_i)$ and $\left ( \bigvee_{i\in I} \varepsilon_i \right ) \otimes \delta =  \bigvee_{i\in I} (\varepsilon_i  \otimes \delta)$.
\end{definition}

\begin{notation}The element $\kappa$ is called the \emph{unit} of the quantale, and $\otimes$ its \emph{tensor} (or multiplication). Besides $\kappa$, we use Greek letters $\alpha, \beta, \gamma, \delta, \varepsilon, \zeta, \eta$, and $\iota$ to denote elements of $\Omega$.
The \emph{top} and \emph{bottom} elements of a quantale are denoted by $\top$ and $\bot$,  respectively. 
\end{notation}

\begin{definition}[Lawverean quantales]  
A quantale $\qOmega = (\Omega, \ord, \otimes, \kappa)$ is called \emph{integral} if $\kappa = \top$ and \emph{non-trivial} if $\kappa \neq \bot$. It is commutative if the monoid $(\Omega, \otimes, \kappa)$ is commutative. $\qOmega$ is \emph{cointegral} if $\varepsilon \otimes \delta = \bot$ implies either $\varepsilon = \bot $ or $\delta = \bot $. 
\emph{Lawverean} quantales are commutative, integral, cointegral, and nontrivial.
\end{definition}

\begin{example} Table \ref{tab:freq} brings some concrete quantales. Notice that the fuzzy quantale $\mathbb{I}$ is Lawverean for the G\"odel and product T-norms, but not for the {\L}uka\-siewicz T-norm. 
\begin{table}[!h]
\centering
  \caption{Correspondence between quantales $\qOmega$ (generic), $\mathbb{2}$ (Boolean), $\mathbb{L}$ (Lawvere), $\mathbb{L}^{\max}$ (strong Lawvere), and $\mathbb{I}$ (fuzzy). Examples from \cite{Gavazzo.2022}.
  }
  \label{tab:freq}
  \begin{tabular}{p{1.5cm}w{c}{1.4cm}w{c}{1.4cm}w{c}{1.4cm}w{c}{1.4cm}c}
    \toprule
             &  $\qOmega$ &  $\mathbb{2}$ & $\mathbb{L}$ & $\mathbb{L}^{\max}$   &  $\mathbb{I}$  \\
    \midrule
   Carrier   & $\Omega$  &  $\{0, 1\}$  & $[0,\infty]$   & $[0,\infty]$  & $[0,1]$        \\
   Order     & $\ord$    &  $\leqslant$            & $\geqslant$         & $\geqslant$        & $\leqslant $        \\
   Join      & $\vee$    &  $\exists$         & $\inf$         & $\inf$        & $\sup$         \\
   Meet      & $\wedge$  &  $\forall$         & $\sup$         & $\sup$        & $\inf$         \\
   Tensor    & $\otimes$ &  $\wedge$          & $+$            & $\max$        & left-continuous T-norm   \\
   Unit      & $\kappa$  &  $1$               & $0$            & $0$           & $1$             \\
  \bottomrule 
\end{tabular}\vspace{-3mm}
\end{table}
\end{example}

\begin{remark} In any quantale $\qOmega = (\Omega, \ord, \otimes, \kappa)$:
\begin{enumerate}[label=(\roman*)]
    \item $\otimes$ is monotonous: $\alpha\ord\beta \Rightarrow \alpha\otimes\gamma\ord\beta\otimes\gamma$ and $\gamma \otimes \alpha\ord \gamma \otimes \beta$.
    
    Indeed, $\alpha\otimes\gamma \ord {}  (\alpha\otimes\gamma)\vee(\beta\otimes\gamma)=(\alpha\vee\beta)\otimes\gamma = \beta\otimes\gamma$. 
    \item If $\qOmega$ is integral then $\alpha\otimes\beta\ord\alpha\wedge \beta$.
    
    Using monotonicity and integrality, we obtain $\alpha\otimes\beta\ord\alpha\otimes\top=\alpha$ and similarly, $\alpha \otimes \beta \ord \beta$; thus, $\alpha \otimes \beta \ord \alpha \wedge \beta$.
\end{enumerate}

\end{remark}

In the rest of the paper, we consider Lawverean quantales.

We now define quantitative ternary relations. 
This notion differs from the ``$\qOmega$-ternary relations'' introduced in \cite{Gavazzo.2022}, which are also required to be antitone in the quantale argument, a condition that we do not impose here. 
\begin{definition}[Quantitative ternary relations] \label{def:qtnr}
    Let $\qOmega$ be a Lawverean quantale, and let $A,B$ be sets. 
    $R \subseteq A \times \Omega \times B$ is said to be a \emph{quantitative ternary relation}. 
    The composition of two quantitative ternary relations $R \subseteq A \times \Omega \times B$ and $S \subseteq B \times \Omega \times C$ is the least quantitative ternary relation $(R;S) \subseteq A \times \Omega \times C$ satisfying $(R;S)(a,\varepsilon \otimes \delta,c)$ whenever $R(a,\varepsilon,b)$ and $S(b,\delta,c)$ hold.
   
    For any set $A$, the diagonal $\Delta_{\mathbb{\Omega},A} \subseteq A \times \Omega \times A$ is the defined by $\Delta_{\mathbb{\Omega},A}\coloneqq \{(a,\kappa,a)\mid a \in A\}$. 
    For a quantitative ternary relation $R \subseteq A \times \Omega \times A$ and $n \in \mathbb{N}$, we define the $n$-th power of $R$ inductively by 
    \[
    	R^0 \coloneqq \Delta_{\mathbb{\Omega},A}
    	\qquad \text{and} \qquad 
    	R^{n+1} \coloneqq (R;R^n). 
    \] 
    The transitive closure $R^+$ and the reflexive transitive closure $R^*$ of $R$ are respectively defined by \[
    	R^+ \coloneqq \bigcup_{n \geqslant 1} R^n \qquad \text{and} \qquad R^* \coloneqq R^0 \cup R^+
    .\] 
\end{definition}

%

\begin{definition}[Quantale homomorphism]
    \label{def:hom}
    Given quantales $\qOmega=(\Omega, \ord, \otimes, \kappa)$ and $\qTheta = (\Theta, \sqsubseteq, \boxtimes, \eta)$, a \emph{quantale homomorphism} is a monotone map $h : \Omega \longrightarrow \Theta $ such that
    $h(\kappa) = \eta $, 
    $h(\varepsilon) \boxtimes h(\delta) = h(\varepsilon \otimes \delta)$, and $h\left(\bigvee_i \varepsilon_i \right) = \bigsqcup_i h(\varepsilon_i)$ (where $\bigvee$ and $\bigsqcup$ are joins of $\qOmega$ and $\qTheta$, respectively).
    For notational convenience, we also say that $h$ is a quantale homomorphism $\qOmega \rightarrow \qTheta$. 
\end{definition}

In the remainder of the paper, we work with quantale endomorphisms, i.e. homomorphisms $\qOmega \rightarrow \qOmega$. We denote them by $\phi,\psi,\ldots$. In the literature, they are called \emph{change of base endofunctors}, or CBEs. Notice that the set of CBEs $\qOmega \rightarrow \qOmega$ is closed under composition.

\begin{example}
The identity function $\bfone\colon\Omega \rightarrow \Omega$ on an arbitrary quantale domain $\Omega$ and multiplication by a nonnegative real constant in the case of the Lawvere quantale are trivial examples of CBEs. 
\end{example}
\begin{remark}
For a given quantale $\qOmega=(\Omega, \ord, \otimes, \kappa)$, its order $\ord$ and tensor $\otimes$ can be extended to CBEs pointwise as follows. Consider $\phi$ and $\psi$ CBEs on $\qOmega$:
\begin{itemize}
    \item [i)] $\phi \ord \psi$ iff  $\phi (\alpha) \ord \psi (\alpha)$ for all $\alpha \in \Omega$.
    \item[ii)] $(\phi \otimes \psi)(\alpha)= \phi (\alpha) \otimes \psi (\alpha)$ for all $\alpha \in \Omega$.
\end{itemize}
\end{remark}
\begin{notation}
We denote by $\kappa^\star$ the constant CBE such that $\kappa^\star(\alpha)=\kappa $, for all $\alpha\in \Omega$.
\end{notation}

We assume $\qOmega$ to be fixed and consider a structure $\qPhi=(\Phi, \ord, \circ, \bfone, \otimes, \kappa^\star)$, where $\Phi$ is a set of CBEs containing the identity function $\bfone$ and constant function $\kappa^\star$, closed under composition $\circ$ and tensor $\otimes$.

\section{Graded Equalities and Rewriting Rules}
\label{sect:gradedTRS}
We assume familiarity with term rewriting, unification, and narrowing as presented in textbooks (e.g., \cite{DBLP:books/daglib/0092409,DBLP:books/daglib/0008995}) and papers as \cite{DBLP:journals/aaecc/MiddeldorpH94,Baader.2001}. In this section, we introduce quantitative unification problem and adapt quantitative TRSs from \cite{Gavazzo.2022} for addressing unification using graded narrowing later. 
\begin{definition} [Modal arity, $\qPhi$-graded signature, (ground and linear) terms, positions, and substitutions] 
    \begin{enumerate}
        \item A \emph{$\qPhi$-graded signature} $\cF$ is a set of function symbols, each equipped with a fixed nonnegative modal arity. The \emph{modal arity} of an $n$-ary function symbol $f$ is a tuple $(\phi_1, \ldots , \phi_n)$ with
        $\phi_i\in\Phi$. In this case, we say that $f$ has \emph{sensitivity} (or \emph{modal grade}) $\phi_i$ on its $i$-th argument and use the notation 
        $f:(\phi_1, \ldots , \phi_n)$ to represent it. 
        \item The set of \emph{terms} over a $\qPhi$-graded signature $\cF$ and a set of variables $\cV$ is defined in the standard way and denoted by $T(\cF,\cV)$. 
        The notion of a \emph{position} in a term is also standard. 
        \item A term is \emph{ground} if it contains no variables, and \emph{linear} if it does not contain repeated occurrences of the same variable.
        \item A \emph{substitution} is a map $\sigma\colon \cV\longrightarrow T(\cF,\cV)$ which maps all but finitely many variables to themselves. 
 We use the set notation for substitutions, writing $\sigma$ explicitly as a finite set $\{x\mapsto \sigma(x) \mid x\neq \sigma(x)\} $. 
The \emph{domain} and \emph{range} of $\sigma$ are defined as $\dom(\sigma)\coloneqq \{x \mid x\neq \sigma(x)\}$ and $\ran(\sigma)=\{\sigma(x) \mid x\in\dom(\sigma) \}$, respectively.
A substitution $\sigma$ extends naturally to an endomorphism on $T(\cF,\cV)$. 
The image of a term $t$ under this endomorphism is denoted by $t\sigma$. 
        \end{enumerate}
\end{definition}

\begin{notation}
The set of positions of a term $t$ is denoted $\pos(t)$. 
For two positions $p$ and $q$, we write $p.q$ for their concatenation, and $p \sqsubseteq q$ to denote that $p$ is a prefix of $q$. 
The subterm of a term $t$ at position $p$ is represented by $t|_p$, while $t[s]_p$ is the term obtained from $t$ by replacing $t|_p$ by $s$.
Moreover, we denote the set of variable positions of $t$ by $\vpos(t)$, and the set of non-variable positions of $t$ by $\nvpos(t)$.

Given a term $t\in T(\cF,\cV)$, we denote by $\cV(t)$ the set of variables appearing in $t$.  For a set of terms $S$, $\cV(S)=\bigcup_{t\in S}\cV(t)$. 

The set of ground terms over $\cF$ is denoted by $T(\cF)$.  

Greek letters $\sigma,\rho,\varphi,\vartheta,\tau$ are used for substitutions, while $\Id$ denotes the identity substitution. The set of substitutions is denoted by $\sset$.
For substitutions $\sigma$ and $\tau$, we write $\sigma \leqq \tau$ if $\sigma\rho = \tau$ holds for some substitution $\rho$. 
Similarly, we write $t \leqq s$ for two terms $t$ and $s$ if $t \rho = s$ for some substitution $\rho$. 
\end{notation}

\begin{definition} [Grade of a position  and a variable in a term]
\begin{enumerate}
        \item Given a term $t$ and a position $p$ in it, the \emph{grade $\grdeg_p(t)$ of $p$ in $t$} is a CBE defined as follows:
        \begin{itemize}
            \item $\grdeg_\lambda(t) \coloneqq \bfone$, where $\lambda$ denotes the top position (empty sequence).
            \item $\grdeg_{i.p}(f(t_1,\ldots,t_n)) \coloneqq\phi_i\circ \grdeg_p(t_i)$, where $f:(\phi_1,\ldots,\phi_n)\in \cF$.
        \end{itemize}
        \item The grade of a variable $x$ in a term $t$, denoted by $\grdeg_x(t)$, is defined as:
        \[ \grdeg_x(t)\coloneqq \left\{
        \begin{array}{cc}
           \hspace{-6mm}\kappa^\star,  & \mbox{ if } x \not\in\cV(t)\\ 
           \displaystyle \bigotimes_{\{p \;\mid \;t|_p = x\}}  \grdeg_p(t),  & \mbox{ otherwise}
        \end{array}
        \right.
        \]
    \end{enumerate}
\end{definition}

\begin{definition}[Graded equational theory] 
Consider a quantitative ternary relation given by a set $E$ of triples $(t,\varepsilon,s)$,  where $t,s$ are terms constructed over a $\qPhi$-graded signature $\cF$ and $\varepsilon$ is an element of a Lawverean quantale $\qOmega$.
We denote such triples by $\varepsilon \Vdash t\approx_E s$ and refer to them as $\qOmega$-equalities. 
A \emph{graded quantitative equational theory} (or \emph{graded $(\qOmega,\qPhi)$-equational theory}) is a pair $\cE=(\cF,=_E)$, where $=_E$ is the quantitative ternary relation generated from $\approx_E$ by the rules in Fig.~\ref{fig:qet}\footnote{The original definition in \cite{Gavazzo.2022} includes an Archimedean rule not necessary for our purposes.}.  
We call $E$ a \emph{presentation} of $=_E$. 
\end{definition}
\vspace{-3mm}

\begin{figure}[ht]
    \centering
    \begin{framed}
        \vspace{-2mm}
        \[{\text{\notsotiny \sf (Ax)\ }\frac{\varepsilon \Vdash t\approx_E s}{\varepsilon \Vdash t=_E s}} \quad 
        {\text{\notsotiny \sf (Refl)\ }\frac{}{\kappa\! \Vdash t=_E t} } \quad
        {\text{\notsotiny \sf (Sym)\ }\frac{\varepsilon \Vdash t=_E s}{\varepsilon \Vdash s=_E t}} \quad
        \text{\notsotiny \sf (Trans)\ } \frac{\varepsilon \Vdash t=_E s \ \, \, \, \delta \Vdash s=_E r}{\varepsilon \otimes \delta \Vdash t=_E r} \vspace{2mm}
        \] 
        \[\text{\notsotiny \sf (Ampl)\ }\frac{\varepsilon_1 \Vdash t_1=_E s_1 \quad \cdots \quad \varepsilon_n \Vdash t_n=_E s_n \quad f:(\phi_1,\ldots,\phi_n)\in \cF}{\phi_1(\varepsilon_1) \otimes \cdots \otimes \phi_n(\varepsilon_n) \Vdash f(t_1,\ldots,t_n)=_E f(s_1,\ldots,s_n)}
          \vspace{2mm} \] 
        \[ \text{\notsotiny \sf (Subst)\ }\frac{\varepsilon \Vdash t=_E s}{\varepsilon \Vdash t\sigma =_E s\sigma}\qquad \qquad \text{\notsotiny \sf (Ord)\ }\frac{\varepsilon \Vdash t=_E s \quad \delta \ord \varepsilon}{\delta \Vdash t=_E s}
          \vspace{2mm}
        \]
        \[ \text{\notsotiny \sf (Join)\ }\frac{\varepsilon_1 \Vdash t=_E s \quad \cdots \quad \varepsilon_n \Vdash t=_E s}{\varepsilon_1 \vee \cdots \vee \varepsilon_n \Vdash t=_E s} \]
        \caption{Graded quantitative equational theory}
        \label{fig:qet}
    \end{framed}
\end{figure}

\begin{definition}[Balanced Equational Theory] A graded equational theory is called \emph{balanced} if whenever $\varepsilon \Vdash s \approx t$, we have $\grdeg_x(s)=\grdeg_x(t)$ for any  variable $x$. If all equations in $\approx_E$ are balanced, so are equations in $=_E$.
\end{definition}

The grade of a position in a given context determines how this context amplifies distances: 
\begin{lemma}[Graded Context Rule \qetrulefont{(Context)}] \label{lem:graded_contexts}
	In any graded equational theory $E$, the following inference rule is valid for all terms $u,v,t$, degrees $\varepsilon \in \Omega$ and positions $p \in \pos(t)$: 
	\[
		\qetrulefont{\scriptsize(Context)} \quad \frac{\varepsilon \Vdash u =_E v}{\partial_p(t)(\varepsilon) \Vdash t[u]_p =_E t[v]_p}
	.\] 
\end{lemma}
\begin{proof}
We proceeded by induction on the length of the position $p$.

If $p = \lambda$, the statement holds trivially as $\partial_{\lambda}(t) = \bfone, t[u]_p = u $ and $t[v]_p = v$. 

In the inductive step, consider $p \neq \lambda$. In this case, $t = f(w_1, \ldots, w_n)$. Thus, we can write $p = i.q$ for some position $q$ and $i \in \{1,\dots,n\}$. Also, $t[u]_p = f(w_1, \ldots, w_i[u]_q, \ldots, w_n)$ and $t[v]_p = f(w_1, \ldots, w_i[v]_q, \ldots, w_n)$.
Assume that the statement holds for all positions $|q| < |p|$ and $f: (\phi_1,\dots,\phi_n)$. By induction hypothesis, $\grdeg_q(w_i)(\varepsilon)\Vdash w_i[u]_q =_E w_i[v]_q$. Consequently, by \qetrulefont{(Refl)} and \qetrulefont{(Ampl)} rules, \[
		\phi_i(\grdeg_q(w_i)(\varepsilon)) \Vdash t[u]_p =_E t[v]_p
	\] 
This concludes the proof since $(\phi_i \circ \grdeg_q(w_i))(\varepsilon) = \partial_p(t)(\varepsilon)$, by definition.
\end{proof}

\begin{corollary}\label{cor:GradedReplacementInContext}
Let $u, v$, and $t$ be terms in any graded equational theory $E$. Suppose that $p\in \pos(t)$ and $\kappa \Vdash u =_E v$.  Then,  $\kappa\Vdash t[u]_p =_E t[v]_p$.    
\end{corollary}
\begin{proof}
By Lemma \ref{lem:graded_contexts} one has that $\grdeg_p(t)(\kappa)\Vdash t[u]_p =_E t[v]_p$. The result is obtained since $\grdeg_p(t)$ is a quantale homomorphism. 
\end{proof}

The quantitative unification problem below was introduced in~\cite{DBLP:conf/ijcar/EhlingK24}:
\begin{mdframed}
		\begin{tabularx}{\textwidth}{p{1.5cm}X}
				{\rm \textbf{Given:}}  & A quantale $\mathbb{\Omega}$, $\varepsilon \in \Omega$ (called the threshold) with $\varepsilon \succ \bot$, a finite quantitative ternary relation $E \subseteq T(\cF,\cV) \times \Omega \times T(\cF,\cV)$, and two terms $t$ and $s$. \\
				{\rm \textbf{Find:}}  &  A substitution $\sigma$ such that $\varepsilon \Vdash t\sigma =_E s\sigma$.
		\end{tabularx}
\end{mdframed}
This problem is called the \emph{$(E,\varepsilon)$-unification problem} for $t$ and $s$, denoted $t =^?_{E,\varepsilon} s$ for short. 
If $\sigma$ satisfies $\varepsilon \Vdash t \sigma =_E s \sigma$, then $\sigma$ is called an \emph{$(E,\varepsilon)$-unifier} of $t$ and $s$, or an $E$-unifier of $t$ and $s$ of \emph{degree} $\varepsilon$.  

\begin{definition}[Graded quantitative term rewriting system]
    \label{def:graded:trs:no:weakening}
    An $(\qOmega, \qPhi)$-term rewriting system $((\qOmega, \qPhi)$-TRS, for short) is a pair $\cR = (\cF, R)$ consisting of a $\qPhi$-graded signature $\cF$ and a quantitative ternary relation $R \subseteq T(\cF,\cV) \times \Omega \times T(\cF,\cV)$, where for convenience triples $(l,\varepsilon,r) \in R$ are written as $\varepsilon\Vdash l \mapsto_R r$ and called rewrite rules. 
   Moreover, we assume that $l \not \in \var$ and $\var(r)\subseteq \var(l)$ hold for any rewrite rule $\varepsilon \Vdash l \mapsto_R r$.  
    The rewriting relation $\to_R$ generated by $R$ is defined as the least ternary relation $\to_R$ such that for any terms $s, l,r\in T(\cF,\cV)$, position $p$ in $s$, quantale element $\varepsilon\in \Omega$, and substitution $\sigma$ we have:
    \begin{align*}
        &\frac{\varepsilon \Vdash l\mapsto_R r}{\grdeg_p(s)(\varepsilon) \Vdash s[l\sigma]_p\to_R s[r\sigma]_p},
    \end{align*}
   where infix notation like $\zeta \Vdash t \to_R s$ is used to denote $(t,\zeta,s) \in \; \to_R$. 
\end{definition}

Note that this notion deviates from the one in \cite{Gavazzo.2022}, which includes rules to account for the \qetrulefont{(Ord)}, \qetrulefont{(Join)}, and Archimedean axiom.
To emphasize this deviation visually, we use different arrow symbols, such as $\to$ instead of $\oldto$. 
The rewrite relation defined here preserves a conceptually important property of non-quantitative TRSs, namely that with respect to a finite set of rewrite rules, there are only finitely many single-step rewrites from a given term. 
This is not the case if a structural rule is included to account for \qetrulefont{(Ord)}.

Graded term rewriting systems naturally inherit properties of ordinary TRSs that will be denoted in the standard manner (cf. \cite{Baader.2001}). 
Left- and right-hand sides of the rules are abbreviated as \lhs\ and \rhs, respectively.

\begin{definition} [Linear and ground graded rewriting systems]
\label{def:lin_grs}
An $(\qOmega,\qPhi)$-TRS $\cR$ and its rules are called \emph{left-linear} and \emph{right-linear} if their \lhs's and \rhs's are linear terms, respectively.  
If they are both left- and right-linear,  they are called linear. 
Analogously, $\cR$ is called \emph{left-ground} and \emph{right-ground} if the \lhs's and \rhs's are ground terms, respectively. 
\end{definition}

\begin{notation}
For more notational flexibility, we also allow to write $l \to_{R, \delta} r$ instead of $\delta \Vdash l \to_R r$.
\end{notation}

\begin{definition}[$(R,\varepsilon)$-convertible and $(R,\varepsilon)$-joinable terms]
\label{def:convertibility:joinability}
Given an $(\mathbb{\Omega}, \mathbb{\Phi})$-TRS $\cR=(\cF,\mapsto_R)$, we define quantitative ternary relations $\gets_R$ and $\leftrightarrow_R$ by setting $t \gets_{R, \varepsilon} s$ iff $s \to_{R, \varepsilon} t$ and $\leftrightarrow_R \ \coloneqq \ \to_R \cup \gets_R$. 
If two terms $t,s$ for some $\varepsilon \in \Omega$ satisfy $t \leftrightarrow^*_{R, \varepsilon} s$, then we say that $t$ and $s$ are \emph{$(R,\varepsilon)$-convertible}.
We say that $t$ and $s$ are \emph{$(R,\varepsilon)$-joinable}, denoted by $t \downarrow_{R,\varepsilon} s$, if there exist $\delta_1, \delta_2 \in \Omega$ and a term $r$ such that $t \to^*_{R, \delta_1} r$, $s \to^*_{R, \delta_2} r$, and $\delta_1 \otimes \delta_2 = \varepsilon$. 
\end{definition}
For a given $\qPhi$-graded signature $\cF$ and an $(\qOmega,\qPhi)$-TRS $\cR$, we extend $\mathcal{F}$ to the signature $\mathcal{F}'$ by adding two new function symbols: a constant symbol {\true} and a binary symbol $=^?$ with arity $(\mathbb{1}, \mathbb{1})$. 
    Moreover, we extend the given rewrite relation $R$ to $R'$ by adding the rule $\kappa \Vdash x =^? x \mapsto {\true}$. 
\begin{remark}
    We have $t \downarrow_{R,\varepsilon} s$ iff $t =^? s \to^*_{R',\varepsilon} {\true}$.
\end{remark}

The notion of confluence has been extended to $(\mathbb{\Omega}, \mathbb{\Phi})$-TRSs in \cite{Gavazzo.2022}. 
\begin{definition}[Confluence of $(\qOmega,\qPhi)$-TRSs \cite{Gavazzo.2022}]
	An $(\mathbb{\Omega},\mathbb{\Phi})$-TRS $\mathcal{R}$ is said to be \emph{confluent} if for any terms $t, s_1, s_2$ and quantale elements $\varepsilon_1, \varepsilon_2 \in \Omega$ satisfying $t \to^*_{R, \varepsilon_1} s_1$ and $t \to^*_{R, \varepsilon_2} s_2$, there exists a quantale element $\delta \succsim \varepsilon_1 \otimes \varepsilon_2$ such that $s_1$ and $s_2$ are $(R,\delta)$-joinable. 
\end{definition}
A quantitative analogue of the Church-Rosser result has been provided in \cite{Gavazzo.2022}. 
With slight modifications, it holds also with respect to the rewrite relation described in \Cref{def:graded:trs:no:weakening}. 
\begin{lemma}[Quantitative Church-Rosser] \label{lem:CR}
	Let $\mathcal{R}$ be a confluent $(\mathbb{\Omega}, \mathbb{\Phi})$-TRS, and let $t,s$ be terms and $\varepsilon \in \Omega$. 
    If $t \leftrightarrow_{R, \varepsilon}^* s$, then $t \downarrow_{R, \delta} s$ holds for some $\delta \succsim \varepsilon$, and $t\downarrow_{R,\delta} s$ implies $t \leftrightarrow^*_{R,\delta} s$. 
\end{lemma}
\begin{proof}
    Suppose that $t \leftrightarrow_{R,\varepsilon}^n s$. 
    We show that $t \downarrow_{R,\varepsilon} s$ by induction on $n$. 
    The case $n=0$ holds by definition of $\leftrightarrow_R^0$. 
    For the induction step, suppose that the assertion holds up until $n$ and that $t \leftrightarrow_{R,\gamma}^n u \leftrightarrow_{R,\delta} s$, where $\gamma \otimes \delta = \varepsilon$. 
    By induction hypothesis, $t \downarrow_{R,\gamma'} u$ holds for some $\gamma'\succsim \gamma$; that is, there exist $\gamma_1, \gamma_2 \in \Omega$ with $\gamma_1 \otimes \gamma_2 \succsim \gamma'$ and a term $v$ such that  $t \to^*_{R, \gamma_1} v$ and $u \to^*_{R, \gamma_2} v$. 
    
    If $u \gets_{R,\delta} s$, then we obtain $s \to^*_{R,\delta \otimes \gamma_2} v$, and thus, $s \downarrow_{R,\gamma_1 \otimes \gamma_2 \otimes \delta} t$, and $\gamma_1 \otimes \gamma_2 \otimes \delta \succsim \varepsilon$. 

    If $u \to_{R,\delta} s$, then by confluence, there exist $\zeta,\eta \in \Omega$ with $\zeta \otimes \eta \succsim \gamma_2 \otimes \delta$ and a term $w$ such that $v \to_{R, \zeta} w$ and $s \to_{R,\eta} w$. 
    Thus, we obtain $t \to^*_{R, \gamma_1 \otimes \zeta} w$, whence $t \downarrow_{R, \gamma_1 \otimes \zeta \otimes \eta} s$, and we have $\gamma_1 \otimes \zeta \otimes \eta \succsim \gamma_1 \otimes \gamma_2 \otimes \delta \succsim \gamma' \otimes \delta \succsim \varepsilon$. 
   
    The second statement is obvious. 
\end{proof}

The following two lemmas establish the relationship between an $(\qOmega,\qPhi)$-TRS and the corresponding quantitative equational theory.
\begin{lemma}
    Let $\mathcal{R}$ be an $(\qOmega,\qPhi)$-TRS. 
    If $t \leftrightarrow^*_{R,\varepsilon} s$, then $\varepsilon \Vdash s =_{R} t$.
\end{lemma}
\begin{proof}
	We proceed by induction on the number $n$ of steps in $t \leftrightarrow^*_{R,\varepsilon} s$. 
	The case $n=0$ is covered by the \qetrulefont{(Refl)} axiom of quantitative equational reasoning. 
	For the induction step, suppose that we have $t \leftrightarrow^n_{R,\varepsilon_1} u \leftrightarrow_{R,\varepsilon_2} s$, where $\varepsilon_1 \otimes \varepsilon_2 = \varepsilon$.
	Then we have $\varepsilon_1 \Vdash t =_R u$ by induction hypothesis. 
	Since $=_R$ is closed under symmetry, we may assume that $u \to_{R,\varepsilon_2} s$. 
	Then we can write $u = u[l \rho]_p$, $s = u[r \rho]_p$, and $\varepsilon_2 = \partial_{p}(u)(\zeta)$, where $\zeta \Vdash l \mapsto r$ is a fresh variant of a rule in $R$, $p \in \pos(t)$, and $\rho$ is a substitution.  
	By \qetrulefont{(Ax)} and \qetrulefont{(Subst)}, we obtain $\zeta\Vdash l \rho =_R r \rho$; thus, $\varepsilon_1 \Vdash u =_R s$ follows via the \qetrulefont{(Context)} rule from \Cref{lem:graded_contexts}. 
    Hence, we obtain $\varepsilon_1 \otimes \varepsilon_2 \Vdash t =_R s$ via \qetrulefont{(Trans)}.
\end{proof}

\begin{lemma}\label{lem:rewriting:completeness}
	Let $\cR$ be an $(\qOmega,\qPhi)$-TRS.
	If $\varepsilon \Vdash t =_R s$, then there exist $\delta_1,\dots,\delta_n \in \Omega$ such that $\delta_1 \vee \dots \vee \delta_n \succsim \varepsilon$ and $t \leftrightarrow^*_{R,\delta_i} s$ holds for every $i \in \{1,\dots,n\}$.  
\end{lemma}
\begin{proof}
	We proceed by induction of the complexity of the derivation of $\varepsilon \Vdash t =_R s$.
	We distinguish cases based on the last rule used.
	\begin{description}
		\item[\qetrulefont{(Refl):}]
			In this case, we have $\varepsilon = \kappa$ and $t = s$, and $t \leftrightarrow^0_{R,\kappa} t$ holds by definition. 
		\item[\qetrulefont{(Symm):}]
			If \qetrulefont{(Symm)} is the last rule applied, then $\varepsilon \Vdash s =_R t$ has a proof of smaller complexity. 
			By induction hypothesis, there exist $\delta_1,\dots,\delta_n$ such that $\bigvee_{i=1}^n \delta_i \succsim \varepsilon$ and $s \leftrightarrow^*_{R, \delta_i} t$ holds for any $i$. 
			Then we also have $t \leftrightarrow^*_{R, \delta_i} s$ for any $i$, as desired. 
		\item[\qetrulefont{(Trans):}] 
			We can assume that the last step applied was 
			\[
				\frac{\zeta \Vdash t =_R u, \quad \eta \Vdash t =_R r}{\zeta \otimes \eta \Vdash u =_R s}
			,\] 
			where $\zeta \otimes \eta = \varepsilon$. 
            
			By induction hypothesis, there exist $\zeta_1,\dots,\zeta_N$ and $\eta_1,\dots,\eta_M$ such that $\bigvee_{i=1}^N \zeta_i \succsim  \zeta$, $\bigvee_{j=1}^M \eta_j \succsim \eta$, and $t \leftrightarrow^*_{R, \zeta_i}u$ and $u \leftrightarrow^*_{R, \eta_j} s$ hold for each $i \in \{1,\dots,N\}$ and $j \in \{1,\dots,M\}$. 
			Hence, we have $t \leftrightarrow^*_{R,\zeta_i \otimes \eta_j} s$ for each $i \in \{1,\dots,N\}$ and $j \in \{1,\dots,M\}$, which concludes this part of the proof as we have
			\begin{align*}
				\bigvee_{i,j} \zeta_i \otimes \eta_j &= \bigvee_i \left(\bigvee_j \zeta_i \otimes  \eta_j\right) = \bigvee_i \left(\zeta_i \otimes \bigvee_j \eta_j\right) \\
				& \succsim \bigvee_i \left(\zeta_i \otimes \eta\right) = \left( \bigvee_i \zeta_i \right) \otimes \eta \ \succsim\ \zeta \otimes \eta = \varepsilon
			\end{align*}
			by the distributivity laws. 
		\item[\qetrulefont{(Axiom):}] 
			If $(t,\varepsilon,s) \in R$, then we have $t \to_{R,\varepsilon} s$.
		\item[\qetrulefont{(Ampl):}]
			If $\varepsilon \Vdash t =_R s$ was obtained via \qetrulefont{(Ampl)}, then we can write $\varepsilon = \varphi_1(\varepsilon_1) \otimes \dots  \otimes \varphi_n(\varepsilon_n)$, $t = f(t_1,\dots,t_n)$ and $s = f(s_1,\dots,s_n)$ for some $f \in \mathcal{F}$ with modal arity $(\varphi_1,\dots,\varphi_n)$. 
			By induction hypothesis, for each $i$ in $\{1,\dots,n\}$, there exist $\delta_{i,1},\dots,\delta_{i,N_i} \in \Omega$ such that $t_i \leftrightarrow^*_{R, \delta_{i,j}} s_i$ and $\bigvee_j d_{i,j} \succsim \varepsilon$. 
			Then we have 
			\[
				f(t_1,\dots,t_n) \leftrightarrow^*_{R,\bigotimes_{i=1}^n \varphi_i(\delta_{i,j_i})} f(s_1,\dots,s_n)
			\] 
			for any sequence $(j_1,\dots,j_n)$ such that $1 \leqslant j_i\leqslant N_i$ for each $i \in \{1,\dots,n\}$. 

			Moreover, we have 
			\begin{align*}
				& \hspace{-2em}\bigvee_{\substack{(j_1,\dots,j_n):\\ 1 \leqslant j_i \leqslant N_i\text{ for each $i$}}} \bigotimes_{i=1}^n \varphi_i(\delta_i, j_i)\\
				&= \bigvee_{j_1=1}^{N_1} \dots  \bigvee_{j_n=1}^{N_n} \varphi_1(\delta_{1,j_1}) \otimes \dots \otimes \varphi_n(\delta_{n,j_n})\\
				&= \left(\bigvee_{j_1=1}^{N_1}\varphi_1(\delta_{1,j_1})\right) \otimes  \bigvee_{j_2=1}^{N_2} \dots  \bigvee_{j_n=1}^{N_n} \varphi_2(\delta_{2,j_2}) \otimes \dots \otimes \varphi_n(\delta_{n,j_n})\\
				&= \dots = \left(\bigvee_{j_1 = 1}^{N_1} \varphi_1 (\delta_{1,j_1})\right) \otimes \dots \otimes \left(\bigvee_{j_n =1}^{N_n} \varphi_n(\delta_{n,j_n})\right)\\
				&= \varphi_1\left(\bigvee_{j_1=1}^{N_1}\delta_{1,j_1}\right) \otimes \dots \otimes \varphi_n\left(\bigvee_{j_n=1}^{N_n}\delta_{n,j_n}\right) \succsim \varphi_1(\varepsilon_1) \otimes \dots  \otimes \varphi_n(\varepsilon_n) = \varepsilon
			,\end{align*} 
			as desired.
		\item[\qetrulefont{(Subst):}]
			If $\varepsilon \Vdash t =_R s$ was obtained via \qetrulefont{(Subst)}, we may write $t = t'\sigma$ and $s = s'\sigma$, where $\varepsilon \Vdash t' =_R s'$. 
			By induction hypothesis, there exist $\delta_1,\dots,\delta_n \in \Omega$ such that $\bigvee_{i=1}^n \delta_i \succsim \varepsilon$ and $t' \leftrightarrow^*_{R,\delta_i} s'$ for each $i \in \{1,\dots,n\}$. 
			It suffices to show that also $t' \sigma \leftrightarrow^*_{R,\delta_i} s' \sigma$ holds for each $i \in \{1,\dots,n\}$. 
			So fix $i \in \{1,\dots,n\}$; we prove that $t' \sigma \leftrightarrow^*_{R,\delta_i} s'\sigma$ by induction on the length $N$ of the derivation $t' \leftrightarrow^*_{R, \delta_i} s'$.
			The case $N=0$ is obvious; for the induction step, assume a derivation  $t' \leftrightarrow_{R,\alpha} u \leftrightarrow^N_{R,\beta} s'$, where $\alpha \otimes \beta = \delta_i$. 
			Then we have $u \sigma \leftrightarrow^*_{R,\beta} s'\sigma$. 
			Without loss of generality, assume that $t' \to_{R,\alpha} u$. 
			Then there exist $p \in \pos(t')$, a fresh variant $\zeta \Vdash l \mapsto r$ of a rule in $R$, and a substitution $\rho$ such that $t'|_p = l \rho$, $u = t'[r \rho]_p$, and $\alpha = \partial_p(t')(\zeta)$. 
			Then we have $(t'\sigma)|_p = (t'|_p)\sigma = l \rho \sigma$ and $u \sigma = (t'[r \rho]_p) \sigma = t' \sigma[r \rho \sigma]_p$. 
			Thus, we obtain $t'\sigma \to_{R,\partial_p(t'\sigma)(\zeta)} u \sigma$, and $\partial_p(t'\sigma)(\zeta) = \partial_p(t')(\zeta) = \alpha$. 
			Summing up, we have shown $t \to_{R,\alpha} u \sigma \leftrightarrow^*_{R,\beta} s$, as desired.  

			The case where $t' \gets_{R,\alpha} u$ is treated analogously. 
		\item[\qetrulefont{(Ord):}]
			By induction hypothesis, there exist $\delta_1,\dots,\delta_n \in \Omega$ such that $t \leftrightarrow^*{R, \delta_i} s$ for each $i \in \{1,\dots,n\}$ and $\bigvee_{i=1}^n \delta_i \succsim \zeta$, where $\zeta \succsim \varepsilon$; so $\bigvee_{i=1}^n \delta_i \succsim \varepsilon$. 
		\item[\qetrulefont{(Join):}]
			If we have $\varepsilon_1,\dots,\varepsilon_n$ with $\bigvee_{i=1}^n \varepsilon_i \succsim \varepsilon$ and $\varepsilon_i \Vdash t =_R s$ for each $i \in \{1,\dots,n\}$, then by induction hypothesis, for each $i$, there exist $\delta_{i,1},\dots,\delta_{i,n_i}$ such that $\bigvee_{j=1}^{n_i} \delta_{i,j} \succsim \varepsilon_i$ and $t \leftrightarrow^*_{R,\delta_{i,j}} s$ for each $j \in \{1,\dots,n_i\}$. 
			But then we have $\bigvee_{i,j} \delta_{i,j} \succsim \varepsilon$. 
	\end{description}
\end{proof}

\begin{corollary}
    Let $\mathcal{R}$ be a confluent $(\mathbb{\Omega},\mathbb{\Phi})$-TRS and let $t,s \in T(\mathcal{F},\mathcal{V})$ be two terms.
 	If $\varepsilon \Vdash t =_R s$, then there exist $\delta_1,\dots,\delta_n \in  \Omega$ such that $\delta_1 \vee \dots \vee \delta_n \succsim \varepsilon$ and $t \downarrow_{R,\delta_i} s$ holds for every $i \in \{1,\dots,n\}$.
\end{corollary}

\section{Graded Quantitative Narrowing}
\label{sect:narrowing}

\begin{definition}[Graded Narrowing]
    \label{def:narrowing_alt}
    Let $\cR = (\cF, R)$ be an 
    $(\qOmega, \qPhi)$-TRS. 
    The graded narrowing relation $\leadsto_R$ generated by $R$ is defined as the least ternary relation such that for any terms $l,r,s\in T(\cF,\cV)$ and quantale elements $\varepsilon,\delta\in \Omega$ we have:
    \begin{align*}
        &\frac{\varepsilon \Vdash l\mapsto_R r}
        {\grdeg_p(s)(\varepsilon) \Vdash s\leadsto_R (s[r\rho]_p)\sigma} 
    \end{align*}
    where $\rho$ is a renaming substitution,  $s$ does not have a common variable with $l\rho$ and $r\rho$, $p\in\nvpos(s)$,
and $\sigma$ is a syntactic most general unifier of $s|_p$ and $l\rho$.
\end{definition}
In such a case, we also write $s \leadsto_{R, \sigma, \partial_p(s)(\varepsilon)} (s [r \rho]_p)\sigma$. 
If needed, we also allow to specify the position $p$ and the rewrite rule $\varepsilon\Vdash l \mapsto r$ as subscripts of the relation $\leadsto_R$.
For $n \in \mathbb{N}$, we define the iterated narrowing relation $\leadsto_R^n$ as follows: 
\begin{align*}
	t \leadsto^0_{R, \sigma, \varepsilon} s &:\Leftrightarrow t = s, \sigma = \Id,\text{ and } \varepsilon = \kappa;\\
	t \leadsto^1_{R, \sigma, \varepsilon} s &:\Leftrightarrow t \leadsto_{R, \sigma, \varepsilon} s;\\
	t \leadsto^{n+1}_{R, \sigma, \varepsilon} s &:\Leftrightarrow \text{\parbox{.8\textwidth}{there exist a term $u$, substitutions $\sigma_1, \sigma_2$ and $\varepsilon_1, \varepsilon_2 \in \Omega$ such that $t \leadsto^n_{R, \sigma_1, \varepsilon_1} u$ and $u \leadsto_{R, \sigma_2, \varepsilon_2} s$, where $\sigma_1 \sigma_2 = \sigma$ and $\varepsilon_1 \otimes \varepsilon_2 = \varepsilon$.}} 
\end{align*}
Analogous to the quantitative rewrite relation, we define the reflexive transitive closure of $\leadsto_R$ by $\leadsto^*_R \coloneqq \bigcup_{n\geqslant 0} \leadsto_R^n$. 

\begin{example}
	Consider the quantitative rewrite system $\mathcal{R} = (\mathcal{F}, R)$ from \cite{DBLP:journals/pacmpl/GavazzoF23}, where $\mathcal{F} = \{Z:(), S:(\Id),+:(\Id,\Id)\}$ and $R$ is the relation over the Lawvere quantale $\mathbb{L}$ given by
	\[
		\{ 0 \Vdash x + Z \mapsto x,\ 0 \Vdash x + S(y) \mapsto S(x+y),\ 1 \Vdash S(x) \mapsto x\}
	.\] 
    The relation $\mapsto$ defines actual distances by eliminating successor functions. The induced convertibility relation coincides with the Euclidean distance on natural numbers.\footnote{This system can be seen as a quantitative extension of Example 3.2 from \cite{Middeldorp.1994}, used there for illustrating the standard narrowing.}
	
    Using $\cR$, we apply quantitative narrowing to find an ``approximate'' solution of the equation $x+1 = 3x$ (whose precise solution, $\frac{1}{2}$, cannot be expressed as a term over $\mathcal{F}$). 

	To this end, we solve the quantitative unification problem $x + S(Z) =_{R,1}^? x + x + x$. 

	A successful narrowing derivation with respect to the extended relation $R'$ (see the part after \Cref{def:convertibility:joinability}) is given below. 
    At each step, the subterms at the narrowing position is underlined, and the rewrite rule applied is indicated in parentheses. 
    Here we denote the rules of $R$ by $r_1$, $r_2$, and $r_3$ (in the order that they appear above), and the additional rule $0 \Vdash x =^? x \mapsto \textsc{true}$ of $R'$ by $r_4$. 
	\begin{align*}
			&\underline{x + S(Z)} =^? (x+x)+x & \\
			& \leadsto_{R',\Id,0}\, S(\underline{x + Z}) =^? (x+x)+x & (\text{by }r_2)\\
			& \leadsto_{R',\Id,0}\, \underline{S(x)} =^? (x+x)+x & (\text{by } r_1)\\
			& \leadsto_{R',\Id,1}\, x =^? \underline{(x+x)+x} & (\text{by } r_3)\\
			& \leadsto_{R',\{x\mapsto Z\},0}\, Z =^? \underline{Z+Z} &  (\text{by } r_1)\\
            &\leadsto_{R',\Id,0} \, \underline{Z =^? Z} & (\text{by } r_1)\\ 
            &\leadsto_{R',\Id,0}\, \textsc{true} & ( \text{by }r_4)
	\end{align*}

	We also show an alternative derivation, leading to a different solution: 
	\begin{align*}
			&\underline{x + S(Z)} =^? (x+x)+x &\\
			& \leadsto_{R',\Id,0}\, S(\underline{x + Z}) =^? (x+x)+x & ( \text{by } r_2)\\
			& \leadsto_{R',\Id,0}\, S(x) =^? \underline{(x+x)+x} & (\text{by } r_1)\\
			& \leadsto_{R',\{x\mapsto S(y)\},0}\, S(S(y)) =^? S(\underline{(S(y)+S(y))+y}) & (\text{by } r_2)\\
			& \leadsto_{R',\{y\mapsto Z\},0}\, S(S(Z)) =^? S(\underline{S(Z)+S(Z)}) & (\text{by } r_1)\\
			& \leadsto_{R',\Id,0}\, S(S(Z)) =^? S(S(\underline{S(Z)+Z})) & (\text{by } r_2)\\
			& \leadsto_{R',\Id,0}\, S(S(Z)) =^? S(S(\underline{S(Z)})) & (\text{by } r_1)\\
			& \leadsto_{R',\Id,1}\, \underline{S(S(Z)) =^? S(S(Z))} & (\text{by }r_3)\\
            &\leadsto_{R',\Id,0}\, \textsc{true} & (\text{by } r_4)\\
	\end{align*}
    That is, we have obtained the approximate solutions $\{x\mapsto Z\}$ and $\{x \mapsto S(Z)\}$. 
\end{example}

The notion of a \emph{basic} narrowing derivation is defined as in \cite{DBLP:journals/aaecc/MiddeldorpH94}. 
\begin{definition}[Basic Graded Narrowing]
	Let 
    \[t_1 \leadsto_{R,\sigma_1,\zeta_1,p_1, \varepsilon_1 \Vdash l_1 \mapsto r_1} \dots \leadsto_{R,\sigma_{n-1},\zeta_{n-1},p_{n-1}, \varepsilon_{n-1} \Vdash l_{n-1} \mapsto r_{n-1}} t_n\] 
    be a narrowing derivation. 
	We inductively define the sets of \emph{basic positions} of $t_1,\dots,t_n$ with respect to the above derivation by
	\begin{align*}
		B_1 &\coloneqq \nvpos(t_1), &
		B_{i+1} &\coloneqq \{q \in B_i \mid p_i \not\sqsubseteq q\} \cup \{p_i . q \mid q \in \nvpos(r_i)\}
	.\end{align*} 
	When the narrowing derivation is clear from the context, we refer to $B_i$ as the set of basic positions of $t_i$, denoted $\bpos(t_i)$. 
	We say that the above narrowing derivation is basic if $p_i \in B_i$ holds for each $i \in \{1,\dots,n-1\}$. 
\end{definition}
Requiring that our narrowing derivations are basic means that narrowing may not be applied to terms that have been introduced by a narrowing substitution in a previous step. 
This requirement prohibits some non-terminating derivations while maintaining completeness in the non-quantitative setting \cite{Middeldorp.1994}.
\begin{remark}\label{rem:rightground:basic}
	If $t$ is a linear term and $\cR$ is right-ground, then every narrowing derivation $t \leadsto^*_{R, \varepsilon} s$ is basic. 

	Indeed, we can write the narrowing derivation as
	\[
		t = t_1 \leadsto_{R, \sigma_1,\zeta_1, \varepsilon_1 \Vdash l_1 \mapsto r_1, p_1} \dots \leadsto_{R, \sigma_{n-1}, \zeta_{n-1},\varepsilon_{n-1}\Vdash l_{n-1} \mapsto  r_{n-1}, p_{n-1}} t_n = s
	.\] 
	We will prove the claim by showing that $\bpos(t_i) = \nvpos(t_i)$ for $i=1,\dots,n-1$. 
	We proceed by induction on $i$. 
	The case $i=1$ holds by definition. 

	For the induction step, suppose that we have $\bpos(t_i) = \nvpos(t_i)$. 
	Note that we have $t_{i+1} = t_i[r_i]_{p_i}\sigma_i = t_i[r_i \sigma_i]_{p_i} = t_i[r_i]_{p_i}$, where the second equality holds since by linearity of $t$, the variables in $\dom(\sigma_i)$ do not appear in $t_i$ except possibly below position $p_i$, and the last equality holds as $r_i$ is ground. 
	Thus, we have 
	\begin{align*}
		\bpos(t_{i+1}) &= \{q \in \bpos(t_i) \mid p_i \not\sqsubseteq q\}\cup \{p_i \cdot q \mid q \in \nvpos(r_i)\}\\
			       &= \{q \in \nvpos(t_i) \mid p_i \not\sqsubseteq q\}\cup \{p_i \cdot q \mid q \in \nvpos(r_i)\}\\
			       &= \nvpos(t_i[r_i]_{p_i})  = \nvpos(t_{i+1}). 
	\end{align*}
\end{remark}

\subsection{Quantitative Narrowing Rules: the \narrow\  Calculus}
\label{subsect:narrowingrules}

In the following, we propose a set of transformation rules for $E$-unification ({\narrow} calculus), inspired by basic narrowing \cite{Hullot1980 ,Baader.2001,DBLP:conf/alp/MiddeldorpH92}.
We operate on configurations of the form $e;C;\sigma;\delta$, where: 
\begin{itemize}
	\item $e$ is a term over $\mathcal{F}'$ (the current state of the problem),
	\item $C$ is a set of equations (the current set of constraints), 
	\item $\sigma$ is a substitution (the solution computed thus far),
	\item $\delta$ is an element of $\Omega$ (the current degree of approximation). 
\end{itemize}
The inference rules of the calculus {\narrow} can be seen in Table~\ref{tab:qnarrowInferenceRules}. \vspace{2mm}

\fbox{
\begin{minipage}{.8\textwidth}
\captionof{table}{\narrow\ inference rules}\label{tab:qnarrowInferenceRules}
\vspace{-6mm}

\begin{align*}
    & \rulefont{LP:}\ \textbf{Lazy Paramodulation}\\*[.1em]
	& \quad e[t]_p; C; \sigma; \delta \Longrightarrow_{\grdeg_p(e)(\varepsilon)} e[r]_p; \{l \sigma = t \sigma\} \cup C; \sigma; \delta \otimes \grdeg_p(e)(\varepsilon), \\*
	& \parbox{\textwidth}{where $e\neq{\true}$, $p$ is a non-variable position of $e$, and $\varepsilon \Vdash l \mapsto r$ is a \textit{fresh} variant of a rule in $R$.}\\[.3em]
	& \rulefont{SU:}\ \textbf{Syntactic Unification}\\*[.1em]
	& \quad e; C; \sigma; \delta \Longrightarrow_{\kappa} e; \emptyset; \sigma \rho; \delta, \\*
	& \parbox{\textwidth}{where $C \neq \emptyset$ and $\rho$ is a most general unifier of $C$.}\\[.3em]
    & \rulefont{Cla:}\ \textbf{Clash}\\*[.1em]
	& \quad e; C; \sigma; \delta \Longrightarrow_{\kappa} \textbf{F},\quad
	\text{if $C$ is not unifiable.}\\[.3em]
	& \rulefont{Con:}\ \textbf{Constrain}\\*[.1em]
	& \quad e; C; \sigma; \delta \Longrightarrow_{\kappa} {\true}; C \cup \{e \sigma\}; \sigma; \delta, \quad
    \text{if $e\neq {\true}$.}
\end{align*}
\end{minipage}
}
\vspace{2ex}

\begin{remark}\label{remark:idemp_substitutions}
Applying the \narrow\ inference rules must guarantee that the substitution built in the configuration, say $\sigma$, is well-defined and idempotent.  
To preserve this property in formalizations, the notion of \mgu\ is adapted so that the $\rho$, computed in the application of the rule (\rulefont{SU}) satisfies 
$\dom(\rho)\subseteq\cV(\ran(\sigma)) \mbox{ and } \cV(\ran(\rho))$ is a set of fresh variables regarding the set of all variables used in the whole process.\footnote{It can be achieved, e.g., by defining substitution comparison with respect to a fixed set of variables, which in case of unifiers is the set of variables of the unification problem.} This condition simplifies formal proofs since substitution $\sigma\rho$ built in the rule (\rulefont{SU}) is easily proved to be idempotent.   

A second important observation justifying the possibility of constructing an \mgu\ $\rho$ such that $\dom(\rho)$ is disjoint from $\cV(\ran(\sigma))$ is that all equations included in the second component of the configuration, say $C$, by the rules (\rulefont{LP}) and (\rulefont{Con}) are instantiated by $\sigma$, implying that these equations have no variables in $\dom(\sigma)$.  Therefore, $\rho$ does not need to instantiate variables in $\dom(\sigma)$.  Finally, the selection of variants of the rules with fresh variables for each application of the rule \rulefont{LP} is also required. Such problems are more apparent in computational frameworks of unification. For instance, \cite{DBLP:journals/jar/AyalaRinconFSKN24} uses a notion of ``nice substitutions'' and of ``correct input renamings'' to formalize the correctness of AC-unification. Also, \cite{DBLP:journals/mscs/Ayala-RinconSFS21} and \cite{DBLP:conf/mkm/AyalaRinconFSKN23} use specialized notions of substitution and a mechanism to ``protect variables''  to verify nominal C-unification and nominal AC-matching, respectively. 
\end{remark}

\begin{remark}
   The {\narrow} rules, as defined, may lead to extensive branching. To mitigate this issue, one can impose restrictions, for example, requiring that in \rulefont{LP}, the terms $l$ and $t$ have the same head symbol, or even that they be unifiable. However, such techniques are beyond the scope of this paper.
\end{remark}

The following lemma confirms that the calculus {\narrow} captures basic narrowing. 
\begin{lemma}\label{lem:narrowing:relation:to:calculus}
	If $t \leadsto^*_{R,\varepsilon,\sigma} s$ is obtained via a basic narrowing derivation, then there exists a derivation $t; \emptyset; \Id; \kappa \Longrightarrow_\varepsilon^* s'; \emptyset; \sigma; \varepsilon$ in {\narrow}, where $s'\sigma = s$ and the set of basic positions of $s$ is precisely the set of non-variable positions of $s'$. 
\end{lemma}
\begin{proof}
	We proceed by induction on the number of steps in the basic narrowing derivation $t \leadsto^*_{R, \varepsilon, \sigma} s$. 
	If $n=0$, then we have $t = s$,  $\varepsilon=\kappa$, $\sigma = \Id$, and $\bpos(t) = \nvpos(t)$; so the empty derivation in {\narrow} works. 

	For the induction step, suppose that the assertion holds up until $n$, and that we are given a basic narrowing derivation, 
	\[
		t \leadsto^n_{R, \varepsilon_1, \sigma_1} u \leadsto_{R, \varepsilon_2, \sigma_2} s
	,\] 
	where $\sigma = \sigma_1 \sigma_2$ and $\varepsilon = \varepsilon_1 \otimes \varepsilon_2$. 
	By induction hypothesis, {\narrow} admits a derivation $t; \emptyset; \Id; \kappa \Longrightarrow_{\varepsilon_1}^* u'; \emptyset; \sigma_1; \varepsilon_1$ such that $u'\sigma_1 = u$ and $\bpos(u) = \nvpos(u')$. 
	Let $p \in \pos(u)$ be the position in which the narrowing step $u \leadsto_{R, \varepsilon_2, \sigma_2} s$ takes place; that is, there exists a fresh variant $\varepsilon_2 \Vdash l \mapsto r$ of a rule in $R$ such that $u|_p$ and $l$ are syntactically unifiable with $\mgu(u|_p,l) = \sigma_2$, and we have $s = u[r]_p \sigma_2$.
	Since we have $l = l \sigma_1$ (as $l$ is fresh) and $u|_p= (u' \sigma_1)|_p = u'|_p \sigma_1$, we can also write $\sigma_2 = \mgu(u'|_p \sigma_1, l \sigma_1)$.
	Since the derivation is basic, we have $p \in \bpos(u) = \nvpos(u')$. 
	
	We thus obtain the derivation 
	\begin{align*}
		u'; \emptyset; \sigma_1; \varepsilon_1 &\Longrightarrow_{\rulefont{LP},\varepsilon_2} u'[r]_p; \{l \sigma_1 =^? u'|_p \sigma_1\}; \sigma_1; \varepsilon_1 \otimes \varepsilon_2\\
        &\Longrightarrow_{\rulefont{SU},\kappa} u'[r]_p;\emptyset;\sigma_1 \sigma_2; \varepsilon_1 \otimes \varepsilon_2
	,\end{align*} 
	and $u'[r]_p \sigma_1 \sigma_2 = u[r]_p \sigma_2 = s$. 
	It remains to verify that $\bpos(s) = \nvpos(u'[r]_p)$. 
	By definition, we have
	\begin{align*}
		\bpos(s) &=  \left\{q \in \bpos(u) \mid p \not \sqsubseteq q\right\} \cup \left\{p.q\mid q \in \nvpos(r)\right\}\\
			 &= \left\{q \in \nvpos(u') \mid p \not\sqsubseteq q\right\} \cup \left\{p.q\mid q \in \nvpos(r)\right\}\\
			 &= \nvpos(u'[r]_p)
	,\end{align*}
	concluding the proof.
\end{proof}

Conversely, it is easy to see that if {\narrow} uses the strategy of applying \rulefont{SU} immediately after an application of \rulefont{LP}, then from a non-failing {\narrow} derivation one can directly obtain the corresponding basic narrowing derivation. 

\subsection{Soundness of \narrow}
\label{subsect:soundness}

\begin{lemma}[Unification of the set of equations by \narrow]\label{lem:unifSetEquations}
If there is a derivation $e; C; \sigma; \delta \ \Longrightarrow^* \ e'; \emptyset;  \sigma'; \delta'$, using the rules from {\narrow}, then for all $s=t\in C$, $ s \sigma' = t\sigma'$. 
\end{lemma}
\begin{proof}
By induction on the length $n$ of \narrow\ derivations. 

Case $n=0$. It holds vacuously since $C=\emptyset$.

Case $n>0$.  Assume the derivation is of the form

\[e; C; \sigma; \delta \ \Longrightarrow e_1;C_1;\sigma_1;\delta_1\Longrightarrow^{n-1} e';\emptyset;\sigma'; \delta_n\]

We proceed by case analysis on the \narrow\ rule applied in the first step of the derivation, and assuming by induction hypothesis that for all $s_i=t_i\in C_1$, $ s_i\sigma' = t_i\sigma'$.

\begin{itemize}
\item (\rulefont{LP}):  $e[t]_p;C;S;\delta\Longrightarrow_{\rulefont{LP}} e[r]_p; \{l\sigma=t\sigma\}\cup C; \sigma; \delta\otimes\grdeg_p(e)(\varepsilon)$.  By induction hypothesis, for all $s_i=t_i\in C$, $ s_i\sigma' = t_i\sigma'$.

\item (\rulefont{SU}): $e;C;\sigma;\delta\Longrightarrow_{\rulefont{SU}}  e;\emptyset;\sigma_1; \delta$, where $\sigma_1=\sigma\rho$ and $\rho = \mgu(C)$.  Thus, for all $s_i=t_i\in C$, $ s_i\sigma' = t_i\sigma'$ since $\sigma'$ is a specialization of the substitution $\sigma_1$, and $ s_i\sigma_1 = t_i\sigma_1$ by Remark \ref{remark:idemp_substitutions}.  

\item (\rulefont{Con}): $e;C;\sigma;\delta\Longrightarrow_{\rulefont{Con}}  {\true}; C\cup\{e\sigma\}; \sigma;\delta$.  By induction hypothesis, for all $s_i=t_i\in C$, $ s_i\sigma' = t_i\sigma'$.   

\end{itemize}

Application of the rule (\rulefont{Cla}) need not be considered.  
\end{proof}

\begin{theorem}[Soundness of \narrow]
If  $t=^?s; C; \sigma; \delta \ \Longrightarrow_{\varepsilon}^+ \ {\true}; \allowbreak  \emptyset;  \sigma'; \delta'$ is a derivation using the rules from {\narrow}, then
		 $\varepsilon \Vdash t\sigma'=_R s\sigma'$. 
\end{theorem}
\begin{proof}
    For notational convenience, we write $e$ for the term $t=^? s$. 
    For any term $u$ of the form $l =^? r$, we use the notation $\lhs(u)$ and $\rhs(u)$ to refer to $l$ and $r$, respectively. 
    
	The proof is by induction on the length of the derivation, say $n$.  
	Suppose the derivation is of the form

	\[
		\mathfrak{C}_0 \Longrightarrow_{\zeta} \mathfrak{C}_1 \Longrightarrow_\eta^*  \mathfrak{C}_n
	,\]
	where $\zeta \otimes \eta = \varepsilon$, $\mathfrak{C}_0 = e; C; \sigma; \delta$, $\mathfrak{C}_n = {\true}; \emptyset; \sigma'; \delta'$ and $\mathfrak{C}_i = e_i; C_i; \sigma_i; \delta_i$ for $i=0,\dots,n$ ($\sigma=\sigma_0$ and $\sigma_n= \sigma'$).

	In the case $n = 1$, the sole applicable rule is \rulefont{SU}, so we have $\zeta=\kappa$, and $\mathfrak{C}_0$ has to be of the form ${\true}; C; \sigma; \delta$ and $\mathfrak{C}_1 = {\true}; \emptyset; \sigma_1; \delta$, where $\sigma_1=\sigma\rho$. 
	Since $\kappa\Vdash {\true}$, we also have $\delta\Vdash ({\true})\sigma_1$.

	Proceeding to the induction step, let $n \geqslant 1$ and suppose the statement of the lemma holds up until $n-1$. 
	We prove that the lemma holds for $n$.
	Applying the induction hypothesis to the derivation 
	\[
		\mathfrak{C}_1 \Longrightarrow_\eta^+ \mathfrak{C}_{n},
	\]
	we obtain $\eta \Vdash e_1 \sigma'$. 

	We distinguish cases according to the {\narrow} rule yielding the first step $\mathfrak{C}_0 \Longrightarrow_\zeta \mathfrak{C}_1$.

	\begin{itemize}
		\item Case (\rulefont{LP}): $\mathfrak{C}_0 = e[t]_p;C;\sigma;\delta \Longrightarrow_{\partial_p(e)(\varepsilon)} e[r]_p; \{l\sigma = t\sigma\}\cup C ; \sigma; \delta\otimes\partial_p(e)(\varepsilon)$, where $e=e[t]_p$ and $\iota\Vdash l\rightarrow r\in \mathcal{R}$.   
			On the one side, by induction hypothesis,  $\eta \Vdash e[r]_p\sigma'$. 
			Firstly, assuming that $p=1 . q$, this means that $\eta \Vdash \lhs(e)[r]_q\sigma' =_E \rhs(e)\sigma'$. 
			Notice that $ l\sigma' =_E t\sigma'$ by Lemma \ref{lem:unifSetEquations}.
			Then, by applying Corollary \ref{cor:GradedReplacementInContext} and \mbox{\sf (Sym)} in Fig. \ref{fig:qet} one obtains $\kappa\Vdash \lhs(e)[t]_q\sigma' =_E \lhs(e)[l]_q\sigma'$. 
			Since  $\partial_{q}(\lhs(e))(\iota)\Vdash \lhs(e)[l]_q\sigma' =_E \lhs(e)[r]_q\sigma'$, by \mbox{\sf (Trans)} in Fig. \ref{fig:qet}, $\partial_{q}(\lhs(e))(\iota)\Vdash \lhs(e)[t]_q\sigma' = \lhs(e)[r]_q\sigma'$.

			From the above two main graded equations, using \mbox{\sf (Trans)}, we obtain  the graded equation $\partial_q(\lhs(e))(\iota)\otimes \eta \Vdash \lhs(e)[t]_q \sigma' =_E \rhs(e)\sigma'$, abbreviated as $\partial_q(\lhs(e))(\iota)\otimes \eta \Vdash e\sigma'$.
			Additionally, $\partial_q(\lhs(e))(\iota) = \partial_p(e)(\iota)$ since the modal arity of the symbol $=^?$ is $(\mathbb{1}, \mathbb{1})$.
			Therefore, $\partial_p(e)(\iota)\otimes \eta \Vdash e\sigma'$.  

			Secondly, the case in which $p=2. q$ is similar, except that both \mbox{\sf (Trans)} and \mbox{\sf (Sym)} are required in the second step of the above analysis; indeed, what we obtain is that $\eta\Vdash \lhs(e)\sigma' = \rhs(e)[r]_q\sigma'$, by induction hypothesis, and that  $\partial_{q}(\rhs(e))(\iota)\Vdash \rhs(e)[t]_q\sigma' = \rhs(e)[r]_q\sigma'$.
			Because of that, the necessity of using \mbox{\sf (Sym)}.  

		\item Case (\rulefont{SU}):  $\mathfrak{C}_0 = e;C;\sigma;\delta \Longrightarrow_\kappa e; \emptyset ; \sigma\rho; \delta$, for $\rho=\mgu(C)$.
		This case is clear because $\kappa\otimes\zeta = \zeta$ and $\zeta\Vdash e\sigma'$ by induction hypothesis. 
		\item Case (\rulefont{Con}): $\mathfrak{C}_0 = e;C;\sigma;\delta\Longrightarrow_{\kappa}  {\true}; C\cup\{e\sigma\}; \sigma;\delta$. 
			The remainder of the derivation has to consist only of one application of the rule (\rulefont{SU}), such that $\eta = \kappa$, $\mathfrak{C}_n = \mathfrak{C}_2 = {\true};\emptyset; \sigma\rho; \delta$, and $\rho=\mgu(C\cup\{e\sigma\})$. 
			Then, in particular, $\kappa\Vdash e\sigma\rho$.
	\end{itemize}
    
\end{proof}

\begin{corollary}\label{cor:soundness}
	If $t =^? s; \emptyset; \Id; \kappa \ \Longrightarrow_\varepsilon^* \ {\true}; \emptyset; \sigma, \varepsilon$ is a derivation using rules from {\narrow}, then $\sigma$ is an $(R,\varepsilon)$-unifier of $t$ and $s$. 
\end{corollary}

\section{Completeness of {\narrow}}
\label{sect:completeness}

We will now investigate in what sense and under which conditions the completeness of ${\narrow}$ can be granted. 
A summary of the results is provided in \Cref{fig:completeness:summary} at the end of this section. 


A weak form of completeness could be a statement of the following form:
``If $\varepsilon \Vdash t \tau =_R s \tau$, then {\narrow} admits a derivation $t =^? s; \emptyset; \Id; \kappa \Longrightarrow_{\delta}^* {\true}; \emptyset; \sigma; \delta$, where $\delta \succsim \varepsilon$.''
Note that this statement does not imply any relation between the given solution $\tau$ and the computed substitution $\sigma$; we merely know by \Cref{cor:soundness} that $\sigma$ is also a solution of the problem. It can be also rephrased as ``if two terms are unifiable modulo $R$ with degree $\varepsilon$, then {\narrow} can find their $R$-unifier with some degree better than $\varepsilon$''.

However, as the following example demonstrates, even this weak form of completeness cannot hold in general. 
\begin{example} \label{ex:counterexample:cubic}
	Consider $\mathbb{\Omega} = \mathbb{L}$, and let $R$ be given by $\{1\Vdash a \mapsto c, 1 \Vdash b \mapsto c, 1 \Vdash c \mapsto d\}$ (where $a,b,c,d$ are constants). 
	Let $f$ be a ternary function symbol of modal arity $(\mathbb{1},\mathbb{1},\mathbb{1})$ (where $\mathbb{1}\colon \lbrack 0, \infty \rbrack \longrightarrow \lbrack 0, \infty \rbrack$ is the identity map).
    

	The $R$-unification problem $f(x,x,x) =^?_R f(a,b,d)$ has the ``optimal'' solution $\sigma = \{x \mapsto c\}$ with degree $3$.

	Via narrowing, only the solution $\tau = \{x \mapsto d\}$ with degree $4$ can be found (see below), but not $\sigma$ or the solutions $\{x \mapsto a\}$ and $\{x \mapsto b\}$, which also have degree $4$. 
\[\begin{array}{l}
f(x,x,x) =^?_R f(a,b,d);\emptyset\;\Id;0 \Longrightarrow_{\rulefont{LP}}^2 \\[1mm]
f(x,x,x) =^?_R f(c,c,d);\{a=a,b=b\};\Id; 2 \Longrightarrow_{\rulefont{LP}}^2 \\[1mm]
f(x,x,x) =^?_R f(d,d,d);\{a=a,b=b,c=c\}; \Id; 4 \Longrightarrow_{\rulefont{Con}}  \\[1mm]
{\true}; \{f(x,x,x)= f(d,d,d), a=a, b=b, c=c\}; \Id;   4 \Longrightarrow_{\rulefont{SU}}\\[1mm]
{\true}; \emptyset; \{x \mapsto d\}; 4
\end{array}
\]

Also notice that for $R = \{1\Vdash a \mapsto b, 1 \Vdash b \mapsto c\}$, and the problem  $f(x,x,x) =^?_R f(a,b,c)$, \narrow\ will compute only the solution $\{x\mapsto c\}$ with degree $3$ , but $\{x\mapsto b\}$ is a better solution with degree $2$. 
\end{example}
Observing that the incompleteness in this example is related to the three occurrences of the variable $x$ in the terms to be unified, one might consider restricting inputs to contain at most two occurrences of the same variable as a potential solution. However, this is insufficient, as rewrite rules can increase the number of variable occurrences in unification problems: 
\begin{example}
    Consider again $\qOmega=\mathbb{L}$ and the rewrite relation $R' \coloneqq R \cup \{g(x) \mapsto f(x,x,x)\}$, where $R$ is the rewrite relation from \Cref{ex:counterexample:cubic} and $g$ is a symbol of arity $(\mathbb{1})$. 
    Then the linear problem $g(x) =_{R'}^? f(a,b,d)$ exhibits the same behavior as the problem in \Cref{ex:counterexample:cubic}: narrowing computes the solution $\{x \mapsto d\}$ of degree $4$, whereas the solution $\{x \mapsto c\}$ of degree $3$ cannot be found. 
\end{example}
Requiring $\cR$ to be right-linear removes this possibility. 
Indeed, we obtain completeness of ordinary quantitative narrowing for right-linear TRSs and linear problems. 
\begin{theorem}\label{lem:weak:completeness}
	Suppose that $\cR$ is a right-linear $(\mathbb{\Omega},\mathbb{\Phi})$-TRS, $t$ is a linear term, and $t \tau \to_{R,\varepsilon}^* s$.  
	Then there is a narrowing derivation $t \leadsto^*_{R, \sigma, \delta} s'$ such that $\delta \succsim \varepsilon$ and $s' \leqq s$. 
\end{theorem}
\begin{proof}
	We proceed by induction on the number $n$ of rewrite steps employed in $t \tau \to_{R, \varepsilon}^* s$. 
	If $n=0$, then $s = t \tau$ and $\varepsilon = \kappa$; so the narrowing derivation $t \leadsto^0_{R, \Id, \kappa} t$ works. 
	
	For the induction step, assume that the assertion holds up until $n$, and we are given a rewrite sequence $t \tau \to_{R, \varepsilon_1} t_1 \to^{n}_{R, \varepsilon_2} s$, where $\varepsilon_1 \otimes \varepsilon_2 = \varepsilon$. 
	Let $p \in \pos(t \tau)$ be the position where the first rewrite step takes place.
	Then there exists a fresh variant $\zeta \Vdash l \mapsto r$ of a rule in $R$ and a substitution $\rho$ such that $(t \tau)|_p = l \rho$, $t \tau[r \rho]_p = t_1$, and $\varepsilon_1 = \partial_p(t \tau)(\zeta)$. 
    We may assume without loss of generality that $\dom(\rho) = \var(l)$. 
	We distinguish four different cases: 
	\begin{enumerate}[label=(\arabic*)]
		\item $p \in \pos(t)$ and $t|_p \tau = t|_p$. 
			In this case, we have $t_1 = (t \tau)[r \rho]_p = t[r \rho]_p \tau$ as $r \rho \tau = r \rho$ since $R$ is regular. 
			Thus, by induction hypothesis, there is a narrowing derivation $t[r \rho]_p \leadsto^*_{R,\delta,\sigma} s'$ for some $\delta \succsim \varepsilon_2$ and $s'\leqq s$.
			Moreover, we have $e \leadsto_{R, \varepsilon_1, \rho} t[r \rho]_p$, so we obtain $t \leadsto^*_{R, \varepsilon_1 \otimes \delta, \rho \sigma} s'$ (for some $s' \leqq s$), and $\varepsilon_1 \otimes \delta \succsim \varepsilon$ by monotonicity of $\otimes$. 
		\item $p \in \pos(t)$ and $t|_p = y \in \dom(\tau)$. 
			In this case, define a substitution $\tau'$ by $x \tau' \coloneqq x \tau$ whenever $x \neq y$ and $y \tau' = r \rho$. 
			Then we have $t_1 = t \tau'$ as $t$ is linear. 
			Thus, by induction hypothesis, there exists a narrowing derivation $t \leadsto^*_{R, \sigma, \delta} s'$ with $\delta \succsim \varepsilon_2 \succsim \varepsilon_1 \otimes \varepsilon_2 = \varepsilon$ and $s' \leqq s$. 
		\item $p \in \pos(t)$, but $t|_p \not\in \mathcal{V}$. 
			Then $t|_p \tau = (t \tau)_p = l \rho$. 
			Since the variables of $l$ are fresh, this means that $t|_p$ and $l$ are syntactically unifiable; let $\varphi$ denote their mgu.
			Then there is a narrowing step $t \leadsto_{R, \partial_p(t)(\zeta), \varphi} t [r]_p \varphi$; note that we have $\partial_p(t)(\zeta) = \partial_p(t \tau)(\zeta) = \varepsilon_1$.
			Since the substitution $\tau \rho$ is also a syntactic unifier of $t|_p$ and $l$, we can write $\tau \rho = \varphi \chi$ for some substitution $\chi$. 
			So we have $t_1 = t[r]_p \tau \rho = (t[r]_p \varphi) \chi$.
            Moreover, we claim that $t[r]_p\varphi$ is a linear term. 
            Indeed, $t[r]_p$ is linear as $t$ is linear and $r$ is fresh and linear. 
            Moreover, the range of $\varphi|_{\var(r)}=\varphi|_{\var(l)}$ is a set of linear terms that do not contain any variables occurring in $t[r]_p$; thus, $t[r]_p\varphi$ is linear.
   
			Thus, by induction hypothesis, there is a narrowing derivation $t[r]_p \varphi \leadsto^*_{R, \delta, \sigma} s'$ with $\delta \succsim \varepsilon_2$ and $s' \leqq s$, whence we obtain $t \leadsto^*_{R, \varepsilon_1 \otimes \delta, \varphi \sigma} s'$. 
		\item $p \not \in \pos(t)$. 
			In this case, there exists a position $q \in \pos(t)$ with $q \sqsubseteq p$ such that $t|_q = y \in \dom(\tau)$. 
			Write $q.p' = p$. 
			Again, we define a substitution $\tau'$ by $x \tau' \coloneqq x \tau$ for $x \neq y$ and $y \tau' \coloneqq y \tau[r \rho]_{p'}$. 
			Then we have $t_1 = t \tau'$ (as $y$ cannot appear in any other position of $t_1$ than $q$ by linearity).
			Thus, by induction hypothesis, there exists a narrowing derivation $t \leadsto^*_{R, \sigma, \delta} s'$ with $\delta \succsim \varepsilon_2 \succsim \varepsilon$ and $s'\leqq s$. 
	\end{enumerate}
    \todo{Cases (2) and (4) are treated in a very similar fashion, but at the moment I can't figure out how to write them as one case. 
    Is the common condition that $e|_p \tau \neq e|_p$ and $p \not\in\nvpos(e)$?}
\end{proof}

Basic quantitative narrowing, however, is still not complete even when $\cR$ is right-linear and the unification problem is linear or even ground, if $\cR$ is not balanced, as the following example shows:
\begin{example}
	Consider $\mathbb{\Omega} = \mathbb{L}$, the graded signature $\mathcal{F} = \{f:(\lambda x.3\cdot x), g: (\mathbb{1}), \allowbreak a:(), b:() \}$, and the rewrite relation $R = \{0 \Vdash f(x) \mapsto g(x), 1 \Vdash a \mapsto b\}$.

	The problem $f(a) =_R^? g(b)$ is solved by the identity solution with degree $1$, as we have $0 \Vdash f(a) =_R g(a)$ and $1 \Vdash g(a) =_R g(b)$. 

	However, the only basic narrowing derivation is  $f(a) \leadsto_{R,\Id,3} f(b) \leadsto_{R,\{x\mapsto b\},0} g(b)$, computing a degree of $3$. 
\end{example}
In addition to the previous assumptions, we therefore also need to assume that $\cR$ is balanced. 
We conjecture that these assumptions are sufficient to lift an ordinary quantitative narrowing derivation to a basic one. 
\begin{conjecture}
    Let $\cR$ be a confluent, balanced, and right-linear $(\qOmega,\qPhi)$-TRS. 
    If $t \leadsto^*_{R, \sigma, \varepsilon} s$, where $t$ is a linear term, then there exists a basic narrowing derivation $t \leadsto^*_{R', \sigma, \varepsilon} s$.
\end{conjecture}
In the non-quantitative case, the step from narrowing to basic narrowing in a convergent TRS is done by assuming the rewrite sequence was innermost, as innermost rewrite sequences can always be lifted to basic narrowing sequences. 
In the quantitative setting, we cannot proceed in the same way as we cannot grant that an innermost rewrite sequence yields the optimal degree.
\begin{example}
    Consider $\qOmega = \mathbb{L}$, $R = \{0 \Vdash f(a) \mapsto f(b),\ 2 \Vdash a \mapsto b\}$. 
    Then innermost rewriting only gives $f(a)\to_{R,2} f(b)$, which is not the optimal degree. 
\end{example}
If we restrict ourselves to right-ground TRSs and linear problems, then any narrowing sequence is basic by \Cref{rem:rightground:basic}. 
Hence, we obtain completeness under these assumptions. 

\begin{theorem}\label{cor:basic:narrowing:to:calculus} 
    If $t,s$ are terms in $T(\cF,\cV)$ and there is a basic narrowing derivation $t =^? s \leadsto^*_{R',\varepsilon, \sigma} {\true}$, then {\narrow} (w.r.t. $R$) admits a derivation $t=^? s;\allowbreak \emptyset;\Id;\kappa \Longrightarrow^*_\varepsilon {\true};\emptyset;\sigma';\varepsilon$, where $\sigma|_{\var(t,s)} = \sigma'|_{\var(t,s)}$.
\end{theorem}
\begin{proof}
    Since the symbols $=^?$ and {\true} do not appear in $R$, $t$, and $s$, the only narrowing step with respect to the rule $\kappa \Vdash x =^? x \mapsto {\true}$ must have taken place as the last step of the derivation. 
    Thus, we obtain a derivation 
    \[
        t =^? s \leadsto_{R, \varepsilon, \tau}  u =^? v
    \]
    such that $u =^? v$ is syntactically unifiable with a fresh variant of $x =^? x$, and their most general unifier $\rho$ satisfies $\tau\rho = \sigma$. 
    By \Cref{lem:narrowing:relation:to:calculus}, {\narrow} admits a derivation $t =^? s;\emptyset;\Id;\kappa \Longrightarrow^*_{\varepsilon} u' =^? v'; \emptyset; \tau; \varepsilon$, where $u = u'\tau$, $v = v'\tau$. 
    
    Note that the restriction of $\rho$ to $\var(u,v)$ is the most general unifier of $u$ and $v$. 
    Thus, we obtain the derivation 
    \[
        u' =^? v'; \emptyset; \tau \varepsilon \Longrightarrow_{\rulefont{Con},\kappa} {\true}; \{u'\tau =^? v'\tau\}; \sigma'; \varepsilon \Longrightarrow_{\rulefont{SU},\kappa} {\true};\emptyset;\tau\rho|_{\var(u,v)};\varepsilon
    ,\]
    and $\sigma|_{\var(t,s)} = (\tau\rho)|_{\var(t,s)} = (\tau\rho|_{\var(u,v)})|_{\var(t,s)}$. 
\end{proof}

\begin{corollary}\label{cor:weak:completeness:qnarrow:rewriting}
	Let $t =^? s$ be a linear problem, and let $\cR$ be a confluent, right-ground $(\qOmega,\qPhi)$-TRS. 
	If $t \tau \leftrightarrow_{R,\varepsilon}^* s \tau$, then {\narrow} admits a derivation $t =^? s; \allowbreak \emptyset; \Id; \kappa \Longrightarrow^* {\true}; \emptyset; \sigma; \delta$ such that $\delta \succsim \varepsilon$. 
\end{corollary}
\begin{proof}
	By \Cref{lem:CR}, we have $t \tau \downarrow_{R,\zeta} s \tau$ for some $\zeta\succsim\varepsilon$; that is, we have $t \tau =^? s \tau \allowbreak \to^*_{R',\zeta} {\true}$. 
	Thus, we obtain $t =^? s \leadsto^*_{R', \delta} {\true}$, where $\delta \succsim \zeta \succsim \varepsilon$, by \Cref{lem:weak:completeness}. 
	Since $t =^? s$ is linear and $R$ is right-ground, this narrowing derivation is basic by \Cref{rem:rightground:basic}. 
	So by \Cref{cor:basic:narrowing:to:calculus}, {\narrow} admits a derivation $t=^? s; \allowbreak \emptyset; \Id; \kappa \Longrightarrow^* {\true}; \emptyset; \sigma; \delta$. 
\end{proof}

\begin{corollary}\label{cor:weak:completeness:qnarrow:theory}
    Suppose that $\qOmega$ is a Lawverean quantale whose order $\precsim$ is total.
	Let $t =^? s$ be a linear problem, and let $\cR$ be a confluent, right-ground $(\qOmega,\qPhi)$-TRS. 
	If $\varepsilon \Vdash t\tau =_R s \tau$, then {\narrow} admits a derivation $t =^? s; \emptyset; \Id; \kappa \Longrightarrow^* {\true}; \emptyset; \sigma; \delta$ such that $\delta \succsim \varepsilon$. 
\end{corollary}
\begin{proof}
    By \Cref{lem:rewriting:completeness}, there exist $\delta_1,\dots,\delta_n \in \Omega$ such that $\delta_1 \vee\dots\vee\delta_n \succsim \varepsilon$ and $t \tau \leftrightarrow_{R,\delta_i} s \tau$ for $i=1,\dots,n$. 
    Since the order $\precsim$ is total, we must have $\delta_i \succsim \varepsilon$ for some $i$; thus, \Cref{cor:weak:completeness:qnarrow:rewriting} yields the result. 
\end{proof}

The results and conjectures presented in this section are summarized in \Cref{fig:completeness:summary}.
\begin{figure}
    \centering
    \includegraphics[scale=1.2]{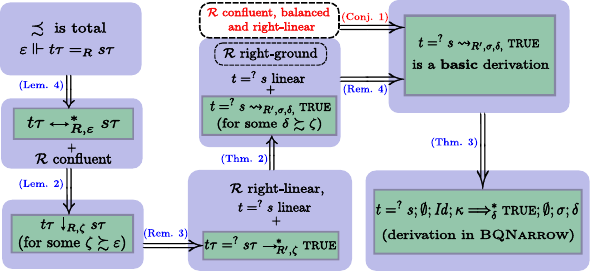}
    \caption{Summary of results and conjectures related to completeness}
    \label{fig:completeness:summary}
\end{figure}

\section{Conclusion}
\label{sect:conclusion}
Narrowing is a well-known technique that generalizes term rewriting by allowing unification rather than mere matching during rule applications. It has a wide range of applications, including equation solving in equational theories, the integration of functional and logic programming, symbolic execution, model checking, reachability analysis, etc.

In this paper, we introduced the concept of narrowing and its basic variant for graded theories, where the quantitative information is specified by (Lawverean) quantales. We applied this framework to quantitative equation solving and established the soundness of the resulting procedure. The quantitative setting presents additional challenges to completeness. To address them, we investigated a restricted class of unification problems and showed that for ordinary quantitative narrowing, completeness holds for linear unification problems in theories induced by confluent and right-linear rewrite systems (with an additional assumption that the underlying quantale is totally ordered). If in these conditions right-linearity is replaced with a stronger right-groundness restriction, we can prove also the completeness of basic narrowing. For future work, we conjectured that completeness of basic narrowing holds for right-linear systems as well, under an additional requirement that the rewrite rules are balanced.

\subsubsection*{Acknowledgments.} This work was partially supported by  the Austrian Science Fund (FWF) project P 35530; the Goiás Research Foundation FAPEG, process  202310267000223; the Brazilian Research Council CNPq, Grant Universal 409003/21-2, and Research Grant 313290/21-0.

\bibliographystyle{splncs04}  
\bibliography{sample-base.bib}

\end{document}